\documentclass[pra,superscriptaddress,twocolumn]{revtex4}
\usepackage{amsfonts,amssymb,amsmath}
\usepackage[]{graphics,graphicx,epsfig}
\usepackage{amsthm}
\usepackage{color}
\usepackage{enumerate}

\usepackage{amsmath,amsthm,amsfonts,amssymb,amscd}
\usepackage[utf8]{inputenc}
\usepackage{indentfirst}
\usepackage{graphicx}
\usepackage[T1]{fontenc}
\usepackage{mathpazo}
\usepackage{changes}

\usepackage[USenglish]{babel}
\usepackage[colorlinks=true]{hyperref}
\usepackage{hyperref}
\usepackage{bm}
\usepackage{lscape}
 \newcommand{\ketbra}[2]{\left|#1\middle\rangle\middle\langle#2\right|}
\newcommand{\de}[1]{\left( #1 \right)}
\newcommand{\De}[1]{\left[#1\right]}
\newcommand{\DE}[1]{\left\{#1\right\}}
\newcommand{\mean}[1]{\langle#1\rangle}
\newcommand{\ie}{\textit{i.e. }}
\def\id{\leavevmode\hbox{\small1\kern-3.8pt\normalsize1}}

\def\identity{\leavevmode\hbox{\small1\kern-3.8pt\normalsize1}}

\renewcommand{\epsilon}{\varepsilon}

\newtheorem{theorem}{Theorem}
\newtheorem{definition}{Definition} 

\def\ba#1\ea{\begin{align}#1\end{align}}
\def\ban#1\ean{\begin{align*}#1\end{align*}}

\newcommand{\be}{\begin{equation}}
\newcommand{\ee}{\end{equation}}

\def\norm#1{ {|\hspace{-.022in}|#1|\hspace{-.022in}|} }

\def\squareforqed{\hbox{\rlap{$\sqcap$}$\sqcup$}}
\def\qed{\ifmmode\squareforqed\else{\unskip\nobreak\hfil
\penalty50\hskip1em\null\nobreak\hfil\squareforqed
\parfillskip=0pt\finalhyphendemerits=0\endgraf}\fi}
\def\endenv{\ifmmode\;\else{\unskip\nobreak\hfil
\penalty50\hskip1em\null\nobreak\hfil\;
\parfillskip=0pt\finalhyphendemerits=0\endgraf}\fi}

\newcommand{\bra}[1]{\langle #1|}
\newcommand{\ket}[1]{|#1\rangle}
\newcommand{\braket}[2]{\langle #1|#2\rangle}

\newcommand{\<}{\langle}
\renewcommand{\>}{\rangle}
\newcommand{\I}{{\rm I}}

\def\id{{\operatorname{id}}}

\DeclareMathOperator{\rank}{rank}

\def\be{\begin{equation}}
\def\ee{\end{equation}}
\def\ben{\begin{eqnarray}}
\def\een{\end{eqnarray}}

\def\bei{\begin{itemize}}
\def\eei{\end{itemize}}

\mathchardef\ordinarycolon\mathcode`\:
\mathcode`\:=\string"8000
\def\vcentcolon{\mathrel{\mathop\ordinarycolon}}
\begingroup \catcode`\:=\active
  \lowercase{\endgroup
  \let :\vcentcolon
  }

\newcommand{\nc}{\newcommand}
 \nc{\proj}[1]{|#1\rangle\!\langle #1 |} 
\nc{\avg}[1]{\langle#1\rangle}

\nc{\conv}{\operatorname{conv}}
\nc{\smfrac}[2]{\mbox{$\frac{#1}{#2}$}} \nc{\Tr}{\operatorname{Tr}}
\nc{\ox}{\otimes} \nc{\dg}{\dagger} \nc{\dn}{\downarrow}
\nc{\lmax}{\lambda_{\text{max}}}
\nc{\lmin}{\lambda_{\text{min}}}

\nc{\csupp}{{\operatorname{csupp}}}
\nc{\qsupp}{{\operatorname{qsupp}}} \nc{\var}{\operatorname{var}}
\nc{\rar}{\rightarrow} \nc{\lrar}{\longrightarrow}
\nc{\poly}{\operatorname{poly}}
\nc{\polylog}{\operatorname{polylog}} \nc{\Lip}{\operatorname{Lip}}
\nc{\Om}{\Omega}
\nc{\wt}[1]{\widetilde{#1}}

\def\>{\rangle}
\def\<{\langle}

\nc{\glneq}{{\raisebox{0.6ex}{$>$}  \hspace*{-1.8ex} \raisebox{-0.6ex}{$<$}}}
\nc{\gleq}{{\raisebox{0.6ex}{$\geq$}\hspace*{-1.8ex} \raisebox{-0.6ex}{$\leq$}}}
\nc{\vholder}[1]{\rule{0pt}{#1}}
\nc{\wh}[1]{\widehat{#1}}
\nc{\h}[1]{\widehat{#1}}
\nc{\ob}[1]{#1}
\def\beq{\begin {equation}}
\def\eeq{\end {equation}}
\def\be{\begin{equation}}
\def\ee{\end{equation}}

\nc{\eq}[1]{(\ref{eq:#1})} 
\nc{\eqs}[2]{\eq{#1} and \eq{#2}}

\nc{\eqn}[1]{Eq.~(\ref{eqn:#1})}
\nc{\eqns}[2]{Eqs.~(\ref{eqn:#1}) and (\ref{eqn:#2})}

\nc{\region}{\cS\cW}

\def\xor{{\sc{xor}}}
\def\bi{\mathcal{B}}

\newenvironment{protocol*}[1]
  {
    \begin{center}
      \hrulefill\\
      \textbf{#1}
  }
  {
    \vspace{-1\baselineskip}
    \hrulefill
    \end{center}
  }

\begin{document}
\title{Quantum bounds on  multiplayer linear games and device-independent witness of genuine tripartite entanglement}
 \author{Gláucia Murta}
 \email{glaucia@fisica.ufmg.br}
\affiliation{Departamento de Física, Universidade Federal de Minas Gerais, Caixa Postal 702, 30123-970, Belo Horizonte, MG, Brazil}
\author{Ravishankar \surname{Ramanathan}}
\affiliation{National Quantum Information Center of Gda\'nsk,  81-824 Sopot, Poland}
\affiliation{Institute of Theoretical Physics and Astrophysics, University of Gda\'{n}sk, 80-952 Gda\'{n}sk, Poland}
 \author{Natália Móller}
\affiliation{Departamento de Física, Universidade Federal de Minas Gerais, Caixa Postal 702, 30123-970, Belo Horizonte, MG, Brazil}
 \author{Marcelo Terra Cunha}
 \affiliation{Departamento de Matemática, Universidade Federal de Minas Gerais,
Caixa Postal 702, 30123-970, Belo Horizonte, MG, Brazil}
\affiliation{Departamento de Matemática Aplicada, Universidade Estadual de Campinas, 13083-970, Campinas, SP, Brazil}
\begin{abstract}
Here we study multiplayer linear games, a natural generalization of XOR games to multiple outcomes.
We generalize a recently proposed efficiently computable bound, in terms of the norm of a game matrix, on the quantum value of 2-player games to 
linear games with $n$ players.
As an example, we bound the quantum value of a generalization of the well-known CHSH game to $n$ players and $d$ outcomes. 
We also {apply the bound to} show in a simple manner that any nontrivial functional box, that could lead to trivialization of communication 
complexity {in a multiparty scenario}, cannot be realized in quantum mechanics. 
We then present a systematic method to derive device-independent witnesses of genuine tripartite entanglement.
\end{abstract}

\maketitle

\section{Introduction}
A nonlocal game is a cooperative task where the players receive questions from a referee and have to give answers
in order to satisfy some previously defined winning condition \cite{NLgames}. 
After the game starts the players are not allowed to communicate, so all they can
do is to previously agree on a joint strategy.
In general the possibility of using a quantum strategy, \ie when the players share an entangled quantum state and perform local measurements
on it, leads to a better performance in the game compared to purely classical strategies
{\ie those that can be described under the paradigm of local hidden variables  {involving the use of shared randomness as a resource}} \cite{Bell}.

Nonlocal games have a wide range of applications. 
They play an important role in the study of communication complexity \cite{BCMdW10,BruknerCommComplexPhys}  (and vice-versa) and in the 
{formulation of device-independent cryptographic protocols \cite{VV14, CR12}}. 
Nonlocal games also constitute a natural framework for studying quantum nonlocality, stressing  the tasks
where quantum resources outperform their classical alternatives.

An important class of nonlocal games comes from the most commonly studied Bell inequalities, namely the correlation Bell
inequalities for binary outcomes otherwise known as XOR games. This is in fact a subclass of the set of \emph{bipartite linear}
games \cite{BavarianShor, xord, Hastad}, defined for an arbitrary number of outputs.
In a linear game \cite{BavarianShor, xord, Hastad}, the parties output answers that are elements of a finite Abelian group and the winning 
constraint depends upon the group operation acting on the outputs. Linear games are the paradigmatic example of nonlocal games with 
two or more outcomes, and a study of their classical and quantum values is crucial, especially in light of applications such as 
the cryptographic primitive task of Bit Commitment \cite{BCJed}.

{For bipartite XOR games,} Tsirelson's theorem \cite{Tsirelson} 
guarantees that the best performance of quantum players can be calculated exactly and efficiently using a
semidefinite program \cite{StephanieSDP,NLgames}.  {The study of XOR games was in part driven by 
the fact that many of the quantum information-processing protocols were developed for qubits, for which binary outcome
games appear naturally. 
Bounds in {\xor} games using game matrices were studied in \cite{xorEpping,xorRavi}.
Recently, there has been much interest in developing applications of higher-dimensional 
entanglement \cite{Exp-high-dim, Qudit-Toffoli, Qudit-randomness, Qudit-key-dist} for which Bell inequalities with more than
two outcomes are naturally suited. Therefore, both for fundamental reasons as well as for these applications, the study of Bell
inequalities with more outcomes is crucial}.

Despite its importance, very few results are known for games with more outputs and/or more parties.
A general result concerning hardness of calculating a tight bound for the probability of success of quantum strategies for
such games was derived in Ref.~\cite{Kempe08}. 
For the special case of \emph{unique} games,  a method of approximating this quantum bound is known \cite{Kempe10},
however the approximation is good only when the quantum probability of success is close to one. 
Recently \cite{xord}, an  {efficiently computable bound for bipartite linear games was derived in terms of the norm of certain game matrices.
These bounds are suitable  {also when} the quantum value is far from unity and they lead to the 
derivation of many interesting results like a generalization of the principle
of no-advantage for nonlocal computation \cite{NLC} to  functions with $d > 2$ outcomes.

As important as the contrast between the performance of classical and quantum players is the fact 
that in some cases quantum strategies cannot  achieve the {limiting value imposed} by the no-signaling principle. 
For example, for the  {class of} linear games there always exists a  no-signaling strategy that wins the game with certainty and this is not always the case for quantum strategies.
The best known case is the CHSH game \cite{CHSH}, with two players and two binary questions per player.
 {In this game,} classical strategies give a maximal probability of success of $\frac34$, quantum strategies reach $\frac{2+\sqrt{2}}{4}$, 
but the so called Popescu-Rohrlich boxes \cite{prbox} are no-signaling devices which allow success with probability $1$.
A celebrated result is that such hypothetical boxes would lead to trivialization of communication complexity \cite{vanDam}.
This result was generalized for the so called \emph{functional boxes}, which can be understood as perfect no-signalling
strategies for some linear games with classical bounds strictly smaller than one \cite{PRd}.


Here we study the performance of quantum strategies in multiplayer linear games.
We generalize the methods of Ref. \cite{xord} and
present an {efficiently computable} bound to the quantum value of $n$-player linear games, and we explore {several interesting} 
applications where the 
bounds lead to nontrivial results.
The text is structured as {follows}: in Section \ref{pre} we introduce the main concepts and definitions
that will be used {in} the text. Section \ref{bound} contains our main result: an {efficiently computable} bound to the quantum 
value of $n$-player linear games. Sections \ref{chsh} and \ref{nofunctbox} present applications of our bound.  
{As a paradigmatic example}, we derive an upper bound to the 
quantum value of a $n$-player generalization of the CHSH-$d$ game. We  {also} show
that uniformly distributed functional boxes  {(that trivialize multiparty communication complexity)} 
cannot be realized in quantum theory. In Section \ref{diew} 
we present a systematic method to design device-independent witnesses of
genuine tripartite entanglement for $d$-dimensional 
systems.   As an example of the method, in Section \ref{example}, we analyze a generalization of  {the}
Mermin GHZ paradox {\cite{Mermin}}, showing that one can detect genuine 
tripartite entanglement in a noisy
GHZ state of local dimension $3$ using only $9$  {expectation} values. Finally in Section \ref{conclu} we discuss our results and 
present future directions.
The more technical proofs of the results presented here will be  {deferred} to the Supplemental Material \cite{supmat}. 

\section{Preliminaries}\label{pre}
In this section we introduce some concepts and definitions that will be used in the remaining text.

\begin{definition}
A $n$-player nonlocal game $g_n(V,p)$ is a cooperative task 
where $n$ players $A_1,\ldots,A_n$,  {who} are not allowed to communicate after 
the game starts, receive respectively questions $x_1,\ldots,x_n$, where $x_i \in [Q_i]$  {($[X]$ denotes
a set with $X$ elements)}, chosen
from a probability distribution $p(x_1,\ldots,x_n)$ by a referee. 
 {Upon} receiving question $x_i$ player $A_i$ is supposed to answer $a_i \in [O_i]$. 
 The winning condition of the game is defined by a predicate function $V(a_1,\ldots,a_n|x_1,\ldots,x_n)$ which
 assumes value $1$ to indicate when the players win and value $0$ to indicate when they lose. 
 
The probability of success of the players for a particular strategy is given by
 \begin{align}
  \omega(g_n)= \sum_{\vec{a}\in \vec{O},\vec{x}\in \vec{Q}} p(\vec{x}) V(\vec{a}|\vec{x}) P(\vec{a}|\vec{x})
 \end{align}
 where $\vec{x}=(x_1,\ldots,x_n)$ denotes the input string, $\vec{Q}=[Q_1]\times \ldots \times [Q_n]$  and analogously for the other vectors. 
\end{definition}

The classical (quantum) value of the game is the maximum probability of success optimizing over all possible classical (quantum) strategies. We will denote the 
classical value by $\omega_c(g_n)$ and the quantum value by $\omega_q(g_n)$.
Another important quantity is the no-signaling value of the game, $\omega_{NS}(g_n)$, which is defined in an analogous way with the players being allowed to
apply any no-signaling strategy.

\begin{definition}[Classical value of a nonlocal game]
 The classical value of a nonlocal game, $\omega_c(g_n)$, is the maximum probability with which the players can win the game when they are 
 restricted to classical strategies. The classical value of the game is always obtained by a deterministic strategy 
 \begin{align}
\omega_c(g_n)=\max_{{\DE{D(a_i|x_i)}}}\sum_{\vec{a},\vec{x}} p(\vec{x}) V(\vec{a}|\vec{x}) D(a_1|x_1)\ldots D(a_n|x_n)
 \end{align}
where $\DE{D(a_i|x_i)}$ represents a deterministic probability distribution.
\end{definition}

\begin{definition}[Quantum value of a nonlocal game]
 The quantum value of a nonlocal game, $\omega_q(g_n)$, is the maximum probability of success when the players can apply quantum 
 strategies. A general quantum strategy can be described by the players sharing an $n$-partite pure state $\ket{\psi}$ of arbitrary dimension and
 performing local measurements $\DE{M_{x_i}^{a_i}}$ on it. The quantum value of the game can be written as
  \begin{align}
\omega_q(g_n)=\sup_{\ket{\psi},\{M_{x_i}^{a_i}\}}\sum_{\vec{a},\vec{x}} p(\vec{x}) V(\vec{a}|\vec{x}) \bra{\psi}M_{x_1}^{a_1} \otimes \ldots \otimes M_{x_n}^{a_n}\ket{\psi}.
 \end{align}
\end{definition}

Note that in order to calculate the quantum value one has  {to optimize} over quantum states, measurements and also over the dimension of these 
operators,  {so that} in principle it is not even known if such a quantity can be evaluated.

Bipartite linear games are a particular class of nonlocal games where the outputs $a,b$ are elements of an Abelian group and 
the winning condition of the game is defined by $a+b=f(x,y)$ (\ie $V(a,b|x,y)=1$ if $a+b=f(x,y)$ and $V(a,b|x,y)=0$ otherwise) where $+$ is the associated group operation. Here we study linear games
with $n$-players which are defined as follows:

\begin{definition}
A $n$-player linear game $g_n^{\ell}(G,f,p)$ is a 
nonlocal game where the players answer with elements of an Abelian group 
$a_1,\ldots,a_n \in G$, where $(G,+)$ is an Abelian group with associated operation $+$, and the predicate function $V$ only 
depends on the sum of the players outputs:
\begin{align}
V(\vec{a}|\vec{x})=\begin{cases}
                    1 \;,\;\text{if}\; a_1 + \ldots + a_n=f(x_1,\ldots,x_n)\\
                    0 \;,\; \text{otherwise}
                        \end{cases}
\end{align}
and then the probability of success of a particular strategy is given by
 \begin{align}\label{psuccess}
 \omega(g_n^{\ell})=\sum_{\vec{x}\in \vec{Q}}p(\vec{x})P(a_1+ \ldots + a_n=f(\vec{x})|\vec{x}),
\end{align}
with $\omega_c(g_n^{\ell}), \omega_q(g_n^{\ell})$, and $\omega_{NS}(g_n^{\ell})$ as defined before.
\end{definition}

\section{An efficiently computable bound on the quantum value of multiplayer linear games}\label{bound}

For  {the} sake of clarity, we will sometimes restrict the presentation of the results to 3-player games. However most of the results are straightforwardly
generalized to $n$ players and we will state it whenever  {this} is the case. 

The success probability of a general strategy for a 3-player linear game is given by
\begin{align}\label{wP}
  \omega(g_3^{\ell}) = \sum_{x,y,z} {p(x,y,z)} P(a+b+c=f(x,y,z)|x,y,z).
 \end{align}
By making use of the Fourier transform of Abelian groups \cite{Terras}, we can rewrite the success probability as
\begin{widetext}
\begin{align}\label{wcor}
  \omega(g_3^{\ell}) = \frac{1}{|G|}\de{1+\sum_{x,y,z}\sum_{k \in G \setminus \DE{e}} {p(x,y,z)} \chi_k(f(x,y,z))\mean{A_x^kB_y^kC_z^k}},
 \end{align}
 \end{widetext}
 where $\mean{A_x^kB_y^kC_z^k}$ are \textit{generalized correlators} defined as the 
Fourier transform of the probabilities
 \begin{align}
  \mean{A_x^iB_y^jC_z^k}=\sum_{a,b,c \in G} \bar{\chi}_i(a)\bar{\chi}_j(b)\bar{\chi}_k(c)P(a,b,c|x,y,z),
 \end{align}
 and $\chi_i$ are the characters of the Abelian group associated with the game. 
 {These} are complex numbers satisfying natural relations: reflexivity $\bar{\chi}_i(a)=\chi_i(-a)$, orthogonality
$\displaystyle{\sum_{a \in G}\chi_i(a)\bar{\chi}_j(a)=|G|\,\delta_{i,j}}$, and homomorphism $\chi_i(a)\chi_i(b)=\chi_i(a+b) \; \; \forall a,b \in G$.

Quantum strategies can be described by projective measurements being performed on pure quantum states in a Hilbert space of arbitrary dimension. 
Consider that a particular strategy is given by the 
set of projective measurements $\{M_x^a\}$, $\{M_y^b\}$, $\{M_z^c\}$ performed  {on} the tripartite quantum state $\ket{\psi}$. In  {this}
case we have the association
\begin{align}
 \mean{A_x^iB_y^jC_z^k}=\bra{\psi}A_x^i\otimes B_y^j \otimes C_z^k \ket{\psi},
\end{align}
where the usually non-Hermitian operators $A_x^i$ are defined as
\begin{align}
 A_x^i=\sum_a\bar{\chi}_i(a)M_x^a,
\end{align}
and analogously for $B_y^j$ and $C_z^k$.

 Motivated by expression \eqref{wcor} we can associate to the game  {$g_3^{\ell}$} a set of $|G|$ matrices  {$\Phi_k$ for $k \in G$},
 which carry information about
 the probability distribution with which the referee picks questions and also the winning condition:
 \begin{align}\label{phi3}
 \Phi_k=\sum_{(x,y,z) \in \vec{Q}}p(x,y,z)\chi_k(f(x,y,z))\ketbra{x}{yz}.
\end{align}

Now we are ready to state an upper bound on the quantum value of tripartite linear games.
\begin{theorem}\label{thmnorm3}
 The quantum value of a tripartite linear game, $g_3^{\ell}(G,f,p)$, where players $A, B$, and $C$ receive questions 
 $(x,y,z)\in [Q_1]\times [Q_2] \times [Q_3]$ and answer with elements of an
 Abelian group $(G,+)$, is upper bounded by
 \begin{align}\label{norm3}
  \omega_q(g_3^{\ell}) \leq \frac{1}{|G|}\de{1+\sqrt{Q_1Q_2Q_3}\sum_{k \in G \setminus \DE{e}}\norm{\Phi_k}},
 \end{align}
where  {$\norm{\cdot}$} denotes the maximum singular value of the matrix, $e$ is the identity element of the group G, and
$ \Phi_k$ are the game matrices.
\end{theorem}

The bound can be generalized for $n$-player games in the following way
\begin{theorem}\label{thmnormn}
 Consider an $n$-player linear game,  $g_n^{\ell}(G,f,p)$. Let $S$ be a subset
 of the parties, $S \subset \DE{A_1,...,A_n}$.
 The quantum value of an $n$-player linear game, $g_n^{\ell}(G,f,p)$, is upper bounded by
 \begin{align}\label{normmulti}
  \omega_q(g_n^{\ell}) \leq \min_{S} \frac{1}{|G|}\de{1+\sqrt{Q_1\ldots Q_n}\sum_{k \in G\setminus \DE{e}}\norm{\Phi^S_k}},
 \end{align}
where $\norm{\Phi^S_k}$ denotes the maximum singular value of matrix $\Phi^S_k$, and the game matrices are defined as 
\begin{align}
 \Phi^S_k=\sum_{{\vec{x} \in \vec{Q}_S, \vec{y} \in \vec{Q}_{S^c}}}p(\vec{x},\vec{y})\chi_k(f(\vec{x},\vec{y}))\ketbra{\vec{x}}{\vec{y}}.
\end{align}
$\vec{x} \in {\vec{Q}_S}$ denotes the vector of inputs to the players that belong to set $S$, and $S^c$ is the complement of $S$.
\end{theorem}

In Theorem \ref{thmnormn} each partition $S$ of the set of parties provides an upper bound to the quantum value, the minimum in Eq. \eqref{normmulti}
selects the most restrictive one. 
Note that for the 3-player game we can also chose to write the game matrices with $S=\DE{B}$ or $S=\DE{C}$, which can lead to tighter bounds than the 
one derived from Eq. \eqref{phi3}. 
Proofs of Theorems \ref{thmnorm3} and \ref{thmnormn} {can be} found in the Supplemental Material \cite{supmat}.

A particular class of linear games are the \xor-$d$ games, $g_n^{\oplus}$, where the outputs belong to the group $\mathbb{Z}_d$ 
with associated operation $\oplus_d$ (sum modulo $d$). The probability of success is then given by
\begin{align}\label{psuccessxor}
 \omega(g_n^{\oplus})=\sum_{\vec{x}\in \vec{Q}}p(\vec{x})P(a_1\oplus_d \ldots \oplus_d a_n=f(\vec{x})|\vec{x}).
\end{align}

In the case of a tripartite \xor-$d$ game with $m$ possible inputs per player, the expression \eqref{norm3} reduces to
 \begin{align}\label{xornorm}
  \omega_q(g_3^{\oplus}) \leq \frac{1}{d}\de{1+\sqrt{m^3}\sum_{k=1}^{d-1}\norm{\Phi_k}},
 \end{align}
where $\norm{\Phi_k}$ denotes the maximum singular value of matrix $\Phi_k$, and
\begin{align}
 \Phi_k=\sum_{x,y,z=1}^{m} p(x,y,z) \zeta^{k.f(x,y,z)}\ketbra{x}{yz}
\end{align}
where $\zeta=e^{2\pi i/d}$ {is a $d$-th root of unity}.

\section{$n$-party CHSH-$d$ game}\label{chsh}
The generalization of the CHSH game for $d$ outputs, where $d$ is a prime or power of prime, was considered in Refs. \cite{chshd,Jichshd,Liangchshd}.
Recently a bound for the quantum value of the CHSH-$d$ game was derived in Ref. \cite{BavarianShor}, using information theoretic arguments,
and re-derived in a simple manner in Ref. \cite{xord} using the bounds  {based on norms of game matrices} for the bipartite case.

Here we consider a generalization of the CHSH-$d$ game for $n$-players based on an expression first considered by Svetlichny \cite{Svetlichny} in the context of detecting genuine multipartite nonlocality. 

\begin{definition} The CHSH${_n}$-$d$ game, for $d$ prime or a power of prime, is a
linear game with the winning condition  {given as}
\begin{align}
 a_1 + \ldots + a_n= \sum_{i<j} x_i \cdot x_j
 \end{align}
where addition and multiplication are operations defined over the  {finite} field $\mathbb{F}_d$.
\end{definition}
An explanation of the operations defining the game together with an example for the case of ternary inputs and outputs
is provided in the Supplemental Material \cite{supmat}.

Using Theorem \ref{thmnormn}, we derive upper bounds for the performance of quantum players in the CHSH${_n}$-$d$ game. 
\begin{theorem}\label{chshd}
 The quantum value of the CHSH${_n}$-$d$ game, for $d$ a prime or a power of a prime, obeys
 \begin{align}\label{boundchsh}
\omega_q(\text{CHSH$_n$-d})\leq \frac{1}{d}+\frac{d-1}{d\sqrt{d}}  .
 \end{align}
\end{theorem}

\begin{proof}[Sketch of the proof]
 The proof follows from direct calculation of $\Phi_k\Phi_k^{\dag}$ with the partition $S=\DE{A_1}$ and using the character relations. We obtain
 \begin{align}
  \Phi_k\Phi_k^{\dag}=\frac{1}{d^{n+1}}\I_d
 \end{align}
which implies $\norm{\Phi_k}=\frac{1}{\sqrt{d^{n+1}}}$. Then by applying Theorem \ref{thmnormn} we obtain
the bound.
The detailed proof can be found in the Supplemental Material \cite{supmat}.
\end{proof}

Interestingly the bounds are independent of the number of parties showing that by increasing the number of
players the performance is still limited.
Analysis of Ref. \cite{Liangchshd} indicates that the bound \eqref{boundchsh} is not tight and for the bipartite case it may
correspond to the value obtained for the first level of the SDP hierarchy introduced in Ref. \cite{NPA}.


\section{No quantum realization of functional nonlocal boxes} \label{nofunctbox}
In Ref. \cite{vanDam} it was shown that the possibility of existence of strong correlations known as PR-boxes \cite{prbox} 
would lead to the trivialization of communication complexity, since by sharing sufficient number of PR-boxes, Alice and Bob would
be able to compute any distributed Boolean function with only one bit of communication. 
{Therefore, the conviction that communication complexity is not trivial (\ie, there are some communication complexity problems that are hard) is viewed as 
a partial characterization of the nonlocal correlations that can be obtained by local measurements on entangled quantum particles.}

Later this result was generalized \cite{PRd} to functional boxes, \ie  a generalization of PR-boxes for  $d$ outputs, where
the outputs satisfy $a \oplus_d b =f(x,y)$, with $d$ prime and any  additively inseparable function $f(x,y)$
(\ie $f(x,y)\neq f_1(x)+f_2(y)$).
Any functional box which cannot be simulated classically would also lead to a trivialization of 
communication complexity \cite{PRd}. 
{Furthermore, a generalization to binary outcome multiparty
communication complexity scenarios was also considered in Ref. \cite{BP05}. In the multiparty problem, $n$ parties are each 
given an input $x_i$ and must compute a function $f(\vec{x})$ of their joint inputs with as little communication as possible. 
In Ref. \cite{BP05}, it was shown that if the parties shared a sufficient number of the full correlation (binary outcome) box with input-output
relation given by $\bigoplus_{i} a_i =x_1\cdot \ldots \cdot x_n$, then any $n$-party communication complexity
problem can be solved with only $n-1$ bits of communication (from $n-1$ parties to the first party who then computes the function), 
thus leading to a trivialization. Analogous results for $d$ output $n$-party functional boxes can also be derived, as we prove in Supplemental Material.

{We now} use our bounds to prove that
any uniformly distributed  total function $n$-party \xor-$d$ game, \ie a game where all $n$-{tuples} of inputs appear with probability greater than zero,
cannot be won with probability $1$ by a quantum strategy unless the game is trivial, \ie the classical value is $1$. 
The boxes that win some of these games with probability one correspond to nontrivial functional boxes, 
hence our result excludes in a simple manner the possibility of {quantum} realization of functional boxes that trivialize
communication complexity in the multiparty scenario (see Supplemental Material).

\begin{theorem}\label{totalfunc}
 For a $n$-player $d$ outcome {\sc{xor}} game $g_n^{\oplus}$ with $m$ questions per player and uniform input distribution
 $p(\vec{x})=1/m^n$, $\omega_q(g_n^{\oplus})=1$ iff $\omega_c(g_n^{\oplus})=1$.
\end{theorem}

The proof of Theorem \ref{totalfunc} can be found in Supplemental Material \cite{supmat}.

We also note that a more general result was derived in Ref.~\cite{noextbox}, showing that no nonlocal extremal box of general no-signaling 
theories can be realized within quantum theory. 
{It remains an open question whether all such nonlocal extremal boxes can lead to a trivialization of the communication complexity problem, and 
whether a simple proof (as for \xor-$d$ games above) can be provided for these general boxes as well.}

\section{Device independent witnesses of genuine tripartite entanglement}\label{diew}
As another application of our central result, we now present a systematic way to derive device-independent witnesses for
genuine multipartite entanglement (DIEW) for tripartite systems.

Characterizing entanglement is a very challenging task.
For bipartite systems positive maps that are not completely positive constitute a powerful tool for generating
 simple operational criteria for detecting entanglement in mixed states \cite{Horodecki1996}.
The most celebrated example is the Peres-Horodecki criterion, also known as PPT
 (positive under partial transposition) criterion.
On the other hand the characterization of multipartite entanglement is even more challenging since inequivalent forms of entanglement appear.
For the detection of genuine multipartite entanglement there is no such direct criteria like the PPT-criteria, however a connexion between 
positive maps and witnesses of genuine multipartite entanglement was established in Ref. \cite{gmePositMaps}, where a framework 
to construct witnesses of genuine multipartite entanglement from positive maps is derived. Other criteria 
to detect genuine multipartite entanglement were proposed in Refs. \cite{Marcus1,Marcus2}.
The development of device-independent witnesses brings together with all the advantage of detecting entanglement 
the possibility of performing this task in
a scenario where we do not have full trust in our devices, which has many applications in cryptographic tasks.

DIEWs were introduced in reference \cite{DIEW} where the authors present a tripartite 3-input 2-output Bell inequality which is able to detect genuine
tripartite entanglement in a noisy three qubit GHZ state, $\rho(V)=V\ketbra{GHZ}{GHZ}+(1-V) \frac{\I}{8}$, for parameter $V> 2/3$. 
In Ref.~\cite{DIEWPalVertesi},
Pál and Vértesi present multisetting Bell inequalities that in the limit of infinitely many inputs
are able to detect genuine tripartite entanglement for parameter as low as $2/\pi$, which is the limit value for which there exist a local model
for the noisy GHZ state for full-correlation Bell inequalities. Other examples of DIEWs with binary outcomes can be found in Ref. \cite{egDIEW}.

A biseparable state of three parties  is a state that can be decomposed into the form
\begin{align}\label{biseparable}
 \rho_{\bi}=p_1\rho_A \otimes \rho_{BC}+p_2\rho_B \otimes \rho_{AC}+p_3\rho_C \otimes \rho_{AB},
\end{align}
where $\vec{p} = \de{p_1,p_2,p_3}$ is a probability vector, $\rho_X$ is a one-party density operator and $\rho_{YZ}$ a two-party density operador.
If a tripartite quantum state cannot be decomposed in the form \eqref{biseparable} it is
said to be genuinely tripartite entangled.  

Let us consider that Alice, Bob and Charlie share a biseparable state of the form $\ket{\psi_{\bi}}=\ket{\psi}_{AB} \otimes \ket{\psi}_{C}$ in 
a tripartite Bell scenario.
In that case Alice and Bob can apply a quantum strategy using their shared entangled state but the best Charlie can do is to apply a classical (deterministic) strategy.
For this case their probability of success will be reduced to
\begin{align}
  \omega(g_3^{\ell}) \leq  \frac{1}{|G|} \bigg( & 1 + \sum_{x,y,z}\sum_{k\in G\setminus \DE{e}}p(x,y,z)\\
  & \times \chi_k(f(x,y,z))\chi_k(-c_z)\bra{\psi}A_x^k\otimes B_y^k\ket{\psi}\bigg), \nonumber
 \end{align}
 where $\DE{c_z}$, $c_z \in G$, represents the outputs of the deterministic strategy performed by Charlie.

 Now by making use of the norm bounds we can bound the performance of the players sharing state  $\ket{\psi_{\bi}}$ in a linear game $g_3^{\ell}$:
   \begin{align}\label{normbic}
  \omega_{\bi}^C(g_3^{\ell}) \leq \max_{\DE{c_z}} \frac{1}{|G|}\de{1+\sqrt{Q_1Q_2}\sum_{k \in G \setminus \DE{e}}\norm{\Phi^{\bi}_k(c_z)}},
 \end{align}
where
 \begin{align}
\Phi^{\bi}_k(c_z)=\sum_{x,y} \de{\sum_z p(x,y,z)\chi_k(f(x,y,z)-c_z)}\ketbra{x}{y}.
\end{align}

In Eq.~\eqref{normbic} we have derived an upper bound for the performance of the players in a game $g_3^{\ell}$ when the players share a quantum state which is
biseparable in the partition $AB|C$. 
In the case of Bell inequalities invariant under permutation of the parties it is sufficient to consider Eq.~\eqref{normbic}.
For general Bell inequalities, the biseparable bound that holds for any state of the form given by Eq.~\eqref{biseparable} can be obtained by taking the maximum 
over the partitions:
\begin{align}
  \omega_{\bi}\leq\max_{X} \omega_{\bi}^X,
\end{align}
where $X \in\DE{A,B,C}$.

 In general we have
 \begin{align}
  \omega_c \leq \omega_{\bi} \leq \omega_q
 \end{align}
 and then, for games where $\omega_{\bi} < \omega_q$, by violating the biseparable bound $\omega_{\bi}$ we can certify in a device independent way that Alice, Bob and Charlie
 share a genuine tripartite entangled quantum state. 

It is important to note that here we are not witnessing genuine multipartite nonlocality. 
DIEWs are a weaker condition than Svetilichny inequalities \cite{Svetlichny}.
Svetlichny inequalities were introduced in multipartite Bell scenarios in order
to detect the existence of genuine multipartite nonlocality. The violation of a Svetlichny inequality guarantees that even if some
parties are allowed to perform a joint strategy they are not able to simulate the exhibited multipartite correlations.
The violation of the bound $\omega_{\bi}$ guarantees that 
the parties share a genuinely tripartite entangled state
but not necessarily that they have genuine tripartite 
nonlocality.

\section{Explicit example and numerical results}\label{example}
Now we exemplify our method deriving a DIEW from a 3-input 3-output tripartite Bell inequality.

Inspired by the Mermin inequality \cite{Mermin} for the GHZ paradox, we can consider the following game:
A referee picks question $x, y, z \in \DE{0,1,2}$ with the promise that $x\oplus_3 y \oplus_3 z =0$. 
The players are supposed to give answers $a,b,c \in \DE{0,1,2}$ in order to satisfy
\begin{align}\label{gameghz}
a \oplus_3 b \oplus_3 c=x\cdot y \cdot z \;\;\; \text{s.t.}\;\;\; x\oplus_3 y \oplus_3 z=0
\end{align}
where the operations $\oplus_3$, $\cdot$ represent, respectively, sum and multiplication modulo 3.

Using the method explained in Section \ref{diew} we can derive the bound 
\begin{align}
 \omega_{\bi} \leq 0.896
\end{align}
(with rhs approximated up to the third decimal).
On the other hand, the $GHZ_3$ state defined as
\begin{align}
 \ket{GHZ_3}=\frac{\ket{000}+\ket{111}+\ket{222}}{\sqrt{3}}
\end{align}
can win the game with probability $1$. Hence we have a device-independent witness 
of genuine tripartite entanglement, since the bi-separable bound is strictly smaller than the quantum bound.
The projective measurements that lead to the maximum value $1$ are
explicitly written in the Supplemental Material \cite{supmat}.

This example also stresses the difference between genuine multipartite entanglement and genuine multipartite nonlocality,
since the Svetlichny bound
\cite{Svetlichny} for this game is $1$, \ie this game cannot be used as a witness of genuine 
tripartite nonlocality, despite being a good witness for  genuine tripartite entanglement.

Now we consider the noisy $GHZ_3$ state
\begin{align}
 \rho_3(V)=V\ketbra{GHZ_3}{GHZ_3}+(1-V)\frac{\I}{27}\,.
\end{align}
By using the optimal measurements for the $GHZ_3$ state we are able to witness genuine multipartite entanglement for $V> 0.85$.

The 3-input 2-output witness presented in Ref. \cite{DIEW} is able to detect genuine multipartite entanglement in $\rho_3(V)$ for $V>
0.81$, however the two-outcome DIEW involves the calculation of 18 expected values whereas here only 9 expected values are involved.

It is important to stress that the method presented here is very general and provide an easy and direct way to find
bi-separable bounds for any 3-player linear game. This opens the possibility for the search of good witnesses of genuine 
multipartite entanglement for high-dimensional systems.

\section{Discussions} \label{conclu}
We have presented an  {efficiently computable} bound  {on} the quantum value of  {an} $n$-player linear game. 
This bound  leads to nontrivial results 
 {and enables proofs of several interesting properties of these games}.
As an application we {have proved an} upper bound on the quantum value of a multipartite generalization of the CHSH-$d$ game, 
 {analogous to the {well-known expression introduced by Svetlichny}. 
 We have also used the bound to show that quantum 
mechanics cannot realize multipartite functional boxes {that win nontrivial linear games and lead to a} trivialization of 
communication complexity.
Finally, we  {have presented} a systematic way to derive device independent witnesses of genuine multipartite entanglement.
The method is very general and can be applied to any tripartite linear game in order to derive tripartite DIEWs with many outcomes.
We  {exhibited} an example were a DIEW involving only $9$ expected values is able to detect genuine tripartite entanglement
in a  {noisy} $GHZ_3$ state.
It remains  {an} open point  {whether these} DIEWs with $d$ outcomes are optimal in terms of number of inputs to detect genuine tripartite
entanglement of systems of dimension $d$.
The search for optimal witnesses with few inputs per player would lead to feasible  {applications} in experiments.

{An interesting question for further research is to investigate the class of functions where the bound on linear games proposed
here is saturated by a quantum strategy. It would also be important to use the bound to identify classes of functions where the 
optimal quantum strategy does not outperform the optimal classical strategy for a linear game, these functions such as those corresponding 
to a distributed nonlocal computation \cite{NLC} can be used to partially characterize the set of quantum correlations. The use of methods based
on Fourier transforms on finite Abelian groups is ubiquitous in communication complexity theory, an important question is to investigate the 
relation between the bounds on linear games derived here and the communication complexity of the associated functions.}
 {Finally,} the bounds derived in Theorem \ref{thmnormn} involves the norm of a matrix which is an object with an intrinsic bipartite structure. 
 A possible future direction would be to  {explore the use of tensors which have} a natural multipartite structure in order to describe the game. 

\textit{Acknowledgements.}
We thank Leonardo Guerini, Cristhiano Duarte, Flavien Hirsch, Mateus Araújo, Ana Cristina Vieira, Marcus Huber, and  Remigiusz Augusiak for useful discussions.
This work was  {supported} by Fundação de Amparo à Pesquisa do Estado de Minas Gerais (FAPEMIG), NCN grant 2013/08/M/ST2/00626, 
Conselho Nacional de Desenvolvimento Científico e Tecnológico (CNPq),  Coordenação de Aperfeiçoamento de Pessoal de Nível Superior (CAPES).
{R.R. is supported by the ERC AdG grant QOLAPS and the Foundation for Polish Science TEAM project co-financed by the EU European Regional Development Fund.}

\appendix
\onecolumngrid

\section*{Supplemental Material: Quantum bounds on  multiplayer linear games and device-independent witness of genuine tripartite entanglement}

\twocolumngrid

\section{An efficiently computable bound on the quantum value of multiplayer linear games}

We have stated that we can associate to the game $ g_3^{\ell}$ a set of $|G|$ matrices  {$\Phi_k$ for $k \in G$},
 which carry all the information necessary to describe the game:
 the probability distribution with which the referee picks questions and also the winning condition. The game matrices are defined as
 \begin{align}
 \Phi_k=\sum_{(x,y,z) \in \vec{Q}}p(x,y,z)\chi_k(f(x,y,z))\ketbra{x}{yz}.
\end{align}

And we make use of these matrices to state our main result

\begin{theorem}\label{sthmnorm3}
 The quantum value of a tripartite linear game, $g_3^{\ell}(G,f,p)$, where players $A, B$ and $C$ receive questions 
 $(x,y,z)\in [Q_1]\times [Q_2] \times [Q_3]$ and answer with elements of an
 Abelian group $(G,+)$, is upper bounded by
 \begin{align}\label{snorm3}
  \omega_q(g_3^{\ell}) \leq \frac{1}{|G|}\de{1+\sqrt{Q_1Q_2Q_3}\sum_{k \in G \setminus \DE{e}}\norm{\Phi_k}},
 \end{align}
where  {$\norm{\cdot}$} denotes the maximum singular value of the matrix, $e$ is the identity element of the group G, and
$ \Phi_k$ are the game matrices.
\end{theorem}

\begin{widetext}
\begin{proof}
 In order to prove the theorem we first define the generalized correlators via the Fourier transform of the 
 probabilities
 \begin{align}
  \mean{A_x^iB_y^jC_z^k}=\sum_{a,b,c \in G} \bar{\chi}_i(a)\bar{\chi}_j(b)\bar{\chi}_k(c)P(a,b,c|x,y,z)
 \end{align}
 where $\chi_j$ are the characters of the abelian group $(G,+)$.

A straightforward calculation of $P(a+b+c=f(x,y,z)|x,y,z)$, using the characters relations  $\bar{\chi}_i(a)=\chi_i(-a)$, 
$\displaystyle{\sum_{a \in G}\chi_i(a)\bar{\chi}_j(a)=|G|\,\delta_{i,j}}$, and $\chi_i(a)\chi_i(b)=\chi_i(a+b)$, gives us
\begin{align}\label{sw}
  \omega(g_3^{\ell}) = \frac{1}{|G|}\de{1+\sum_{x,y,z}\sum_{k\in G\setminus \DE{e}}p(x,y,z)\chi_k(f(x,y,z))\mean{A_x^kB_y^kC_z^k}}
 \end{align}
 and for a quantum strategy where measurements $\{M_x^a\}$, $\{M_y^b\}$, $\{M_z^c\}$ are 
 performed on the tripartite quantum state $\ket{\psi}$ we have
 \begin{align}\label{swq}
  \omega(g_3^{\ell}) = \frac{1}{|G|}\de{1+\sum_{x,y,z}\sum_{k \in G\setminus \DE{e}}p(x,y,z)\chi_k(f(x,y,z))\bra{\psi}A_x^k\otimes B_y^k \otimes C_z^k\ket{\psi}}.
 \end{align}

  Since the generalized correlators are unitary operators we can define the following normalized vectors
  \begin{subequations}
 \begin{align}
  \ket{\alpha^k}=&\frac{1}{\sqrt{Q_1}}\sum_{x \in [Q_1]}A_x^k\otimes \I_{BC} \otimes \I_{Q_1} \ket{\psi}\ket{x}\\
  \ket{\beta^k}=& \frac{1}{\sqrt{Q_2Q_3}}\sum_{x,y \in [Q_2]\times [Q_3]} \I_{A}\otimes B_y^k \otimes C_z^k  \otimes \I_{Q_1,Q_2} \ket{\psi}\ket{y,z}.
 \end{align}
 \end{subequations}

Now using the game matrices defined in the text
 \begin{align}
 \Phi_k=\sum_{(x,y,z) \in \vec{Q}}p(x,y,z)\chi_k(f(x,y,z))\ketbra{x}{yz}.
\end{align}
One have that
 \begin{align}\label{swq3}
  \omega(g_3^{\ell}) &= \frac{1}{|G|}\de{1+\sqrt{Q_1Q_2Q_3}\sum_{k \in G\setminus \DE{e}}\bra{\alpha^k}\I_{ABC}\otimes \Phi_k \ket{\beta^k}}\nonumber\\
  &\leq  \frac{1}{|G|}\de{1+\sqrt{Q_1Q_2Q_3}\sum_{k \in G\setminus \DE{e}}\norm{\I_{ABC}\otimes \Phi_k}}\\
  &=  \frac{1}{|G|}\de{1+\sqrt{Q_1Q_2Q_3}\sum_{k \in G\setminus \DE{e}}\norm{\Phi_k}}\nonumber,
 \end{align}
 where $\norm{\cdot}$ denotes the maximum singular value of the matrix.
 \end{proof}
  \end{widetext}

The bound can be generalized for $n$-players game.

\begin{theorem}\label{sthmnormn}
 Consider an $n$-player linear game,  $g_n^{\ell}(G,f,p)$. Let $S$ be a subset
 of the parties, $S \subset \DE{A_1,...,A_n}$.
 The quantum value of an $n$-player linear game, $g_n^{\ell}(G,f,p)$, is upper bounded by
 \begin{align}\label{snormmulti}
  \omega_q(g_n^{\ell}) \leq \min_{S} \frac{1}{|G|}\de{1+\sqrt{Q_1\ldots Q_n}\sum_{k \in G\setminus \DE{e}}\norm{\Phi^S_k}},
 \end{align}
where $\norm{\Phi^S_k}$ denotes the maximum singular value of matrix $\Phi^S_k$, and the game matrices are defined as 
\begin{align}
 \Phi^S_k=\sum_{{\vec{x} \in \vec{Q}_S, \vec{y} \in \vec{Q}_{S^c}}}p(\vec{x},\vec{y})\chi_k(f(\vec{x},\vec{y}))\ketbra{\vec{x}}{\vec{y}}.
\end{align}
$\vec{x} \in {\vec{Q}_S}$ denotes the vector of inputs to the players that belong to set $S$, and $S^c$ is the complement of $S$.
\end{theorem}

\begin{widetext}
\begin{proof}
 The proof is a straightforward generalization of the tripartite case.
 Considering a $n$-player linear game  $g_n^{\ell}(G,f,p)$, where player $A_i$ receive one among $Q_i$ inputs and outputs $a_i \in G$,
 we can define the generalized correlators as the Fourier transform of the probabilities
 \begin{align}
  \mean{{A_1}_{x_1}^{i_1}\ldots {A_n}_{x_n}^{i_n}}=\sum_{a_1,\ldots, a_n \in G} \bar{\chi}_{i_1}(a_1) \ldots  \bar{\chi}_{i_n}(a_n)P(a_1,\ldots, a_n|x_1,\ldots,n_n).
 \end{align}

 Then analogously to the tripartite case, the probability of success can be written as:
\begin{align}\label{swn}
  \omega(g_n^{\ell}) = \frac{1}{|G|}\de{1+\sum_{x_1,\ldots,x_n}\sum_{k \in G\setminus \DE{e}}p(x_1,\ldots, x_n)\chi_k(f(x_1,\ldots,x_n))\mean{{A_1}_{x_1}^k \ldots {A_n}_{x_n}^k}}.
\end{align}

 Now let us consider $S$ a subset  of the parties $S \subset \DE{A_1,...,A_n}$. Then by using the matrices 
 \begin{align}
 \Phi^S_k=\sum_{\vec{x} \in \vec{Q}_S, \vec{y} \in \vec{Q}_{S^c}}p(\vec{x},\vec{y})\chi_k(f(\vec{x},\vec{y}))\ketbra{\vec{x}}{\vec{y}}
\end{align}
and defining the vectors
\begin{subequations}
  \begin{align}
  \ket{\alpha^k}=&\frac{1}{\sqrt{Q_{S}}}\sum_{\vec{x} \in \vec{Q}_{S}}\de{\otimes_{i \in S}A_{x_i}^k}\otimes \I_{S^c} \otimes \I_{Q_S} \ket{\psi}\ket{\vec{x}}\\
  \ket{\beta^k}=& \frac{1}{\sqrt{Q_{S^c}}}\sum_{\vec{y} \in \vec{Q}_{S^c}} \I_{S}\otimes \de{\otimes_{i \in S^c}A_{x_i}^k}  \otimes \I_{Q_{S^c}} \ket{\psi}\ket{\vec{y}},
 \end{align}
 \end{subequations}
 where $Q_S=\prod_{i}Q_{j_i}$, and $\vec{Q_S}=[Q_{j_1}]\times \ldots \times [Q_{j_k}]$ for $A_{j_i}\in S$.
 
 Now analogously to the tripartite case one has
 
\begin{align}\label{swqn}
  \omega(g_n^{\ell}) &= \frac{1}{|G|}\de{1+\sqrt{Q_1  \ldots Q_n}\sum_{k \in G\setminus \DE{e}}\bra{\alpha^k}\I_{A_1\ldots A_n}\otimes \Phi^S_k \ket{\beta^k}}\nonumber\\
  &\leq  \frac{1}{|G|}\de{1+\sqrt{Q_1  \ldots  Q_n}\sum_{k\in G\setminus \DE{e}}\norm{\I_{A_1\ldots A_n}\otimes \Phi^S_k}}\\
  &=  \frac{1}{|G|}\de{1+\sqrt{Q_1  \ldots  Q_n}\sum_{k\in G\setminus \DE{e}}\norm{\Phi^S_k}}\nonumber.
 \end{align}
 
 By the construction of the proof we see that for all subset $S$ we have a valid upper bound to the quantum value.
 \end{proof}
 \end{widetext}

\section{Finite Fields}\label{sfinitefield}

The generalization of the CHSH game that we consider, the CHSH${_n}$-$d$ game, is defined for elements of a finite field. Hence, for completeness,
in this Section we 
present the definition and properties of the operations in a finite field.

\begin{definition}[Finite Field]
A finite field $\mathbb{F}_d$ is a set of $d$ elements with the operations sum $+$ and multiplication $\cdot$ such that
\begin{enumerate}[(i)]
 \item $a+b$, $a \cdot b$ $\in \mathbb{F}_d$, $\forall a,b \in \mathbb{F}_d$,
 \item $\exists \; 0$ s.t. $a+0=a ,\, \forall a \in \mathbb{F}_d$,
  \item $\exists \; 1$ s.t. $a\cdot 1=a ,\, \forall a \in \mathbb{F}_d$,
  \item $\forall \; a, \exists \; -a$ s.t. $a+(-a)=0$,
  \item $\forall \; b \neq 0, \exists \; b^{-1}$ s.t. $b\cdot b^{-1}=1$,
  \item Associativity: $\forall a,b,c\,\in \mathbb{F}_d$
  \begin{align*}
   a+(b+c)&=(a+b)+c\\
   a\cdot (b\cdot c)&= (a \cdot b) \cdot c
  \end{align*}
\item Commutativity: $\forall a,b\,\in \mathbb{F}_d$
\begin{align*}
   a+b&=b+a\\
   a\cdot b&=  b \cdot a
  \end{align*}
 \item Distributivity: $\forall a,b,c\,\in \mathbb{F}_d$
 $$a\cdot (b+c)= a \cdot b + a \cdot c.$$
\end{enumerate}
\end{definition}

For $d$ prime all the conditions $(i)$-$(viii)$ can be satisfied by arithmetic modulo $d$.
For $d=p^r$ the arithmetic operations are defined by addition and multiplication of polynomials of degree $< r$ over $\mathbb{Z}_p$, $\mathbb{Z}_p\De{X}$. 
In order to construct a field $\mathbb{F}_{p^r}$ one starts by choosing an irreducible polynomial of degree $r$ over  $\mathbb{Z}_p$, this polynomial will
define the zero of the field  by the so called quotient. 

As an example for the field $\mathbb{F}_d$ with $d=2^3$ we can pick the polynomial $X^3+X+1 \in  \mathbb{Z}_2\De{X}$ from 
which we can obtain the relation 
\begin{align}\label{spirred}
X^3+X+1 =0 \Rightarrow X^3=X+1.
\end{align}
Now the elements of the field can be represented by strings $(a,b,c)$, with $a,b,c \in \DE{0,1}$, and we can associate each string with
the polynomial $aX^2+bX+c$. Given that, addition and multiplication will be taken as addition and multiplication of the polynomial reduced by  the 
relation \eqref{spirred}.

\section{$n$-party CHSH-$d$ game.}

The CHSH-$d$ game, for $d$ prime or power of prime, is a generalization of the most remarkable Bell scenario, the CHSH inequality \cite{CHSH}.
In the CHSH-$d$ game, Alice and Bob receive questions $x$ and $y$ chosen uniformly random from a finite field $\mathbb{F}_d$, and they are supposed 
to given answers $a$ and $b$ also from field $\mathbb{F}_d$ in order to satisfy
\begin{align}
 a+b=x\cdot y.
\end{align}

\begin{definition} The $n$-party CHSH-$d$ game, CHSH${_n}$-$d$, for $d$ prime or a power of prime, is a
linear game with winning condition 
\begin{align}
 a_1 + \ldots + a_n= \sum_{i,j>i} x_i \cdot x_j
 \end{align}
where addition and multiplication are operations defined over the field $\mathbb{F}_d$.
\end{definition}

As an example we consider the CHSH${_3}$-$3$, where the inputs and outputs are elements
of $\mathbb{Z}_3$, $a,b,c,x,y,z \in \DE{0,1,2}$, and as we
have seen in Section \ref{sfinitefield} $+,\cdot$ are sum and multiplication modulo 3. 
The characters of the group $(\mathbb{Z}_3,+)$ are $\chi_j=\zeta^j$, where $\zeta=e^{2\pi i/3}$ is the $3^{rd}$ root of unity.

The game matrix $\Phi_1$ for the  CHSH${_3}$-$3$ game will be defined as
\begin{align}
  \Phi_1=\sum_{x,y,z=0}^{2} \frac{1}{27} \zeta^{x\cdot y +x\cdot z+y \cdot z}\ketbra{x}{yz}
\end{align}
which is explicitly given by

\begin{equation}
  \Phi_1=\frac{1}{27} \begin{bmatrix} 
1&1&1&1&\zeta&\zeta^2&1&\zeta^2&\zeta \\
1 &\zeta &\zeta^2 &\zeta&1& \zeta^2&\zeta^2&\zeta^2&\zeta^2 \\
1 & \zeta^2 &\zeta&\zeta^2&\zeta^2  &\zeta^2&\zeta&\zeta^2&1 \end{bmatrix}.
\end{equation}

Now we prove an upper bound to the quantum value of the CHSH${_n}$-$d$ for all $n$ and all $d$ prime or power of prime.

\begin{theorem}\label{schshd}
 The quantum value of the CHSH${_n}$-$d$ game, for $d$ a prime or a power of a prime, obeys
 \begin{align}
\omega_q(\text{CHSH$_n$-$d$})\leq \frac{1}{d}+\frac{d-1}{d\sqrt{d}}  .
 \end{align}
\end{theorem}

\begin{widetext}
\begin{proof}
The proof follows from direct calculation of $\Phi_k\Phi_k^{\dag}$ with the partition $S=A_1$.

For the CHSH${_n}$-$d$ game the game matrices are the following:
\begin{align}
 \Phi_k=\sum_{x_1,\ldots,x_n}\frac{1}{d^n}\chi_k(\sum_{j>i}x_i \cdot x_j)\ketbra{x_1}{x_2 \ldots x_n}.
\end{align}

By making use of the characters 
 relations $\bar{\chi}_i(a)=\chi_i(-a)$, $\displaystyle{\sum_{a \in G}\chi_i(a)\bar{\chi}_j(a)=|G|\,\delta_{i,j}}$ and 
 $\chi_i(a)\chi_i(b)=\chi(a+b)$ we have:
 \begin{align}
   \Phi_k\Phi_k^{\dag}&=\frac{1}{d^{2n}}\sum_{x_1,\ldots,x_n}\sum_{x'_1,\ldots,x'_n}\chi_k(\sum_{j>i}x_i \cdot x_j)\bar{\chi}_k(\sum_{j>i}x'_i \cdot x'_j)\ket{x_1}\braket{x_2 \ldots x_n}{x'_2 \ldots x'_n}\ket{x'_1}\nonumber\\ 
           & =\frac{1}{d^{2n}}\sum_{x_1,\ldots,x_n}\sum_{x'_1}\prod_{j>1}\chi_k(x_1 \cdot x_j)\bar{\chi}_k(x'_1 \cdot x_j)\ketbra{x_1}{x'_1}\\ 
           &=\sum_{x_1}\frac{1}{d^{n+1}}\ketbra{x_1}{x_1}\nonumber\\
           &=\frac{1}{d^{n+1}}\I_d, \nonumber
 \end{align}
which implies $\norm{\Phi_k}=\frac{1}{\sqrt{d^{n+1}}}$ for all $k$. Then by applying Theorem \ref{sthmnormn} we obtain
 \begin{align}
\omega_q(\text{CHSH$_n$-$d$})\leq \frac{1}{d}+\frac{d-1}{d\sqrt{d}}.
 \end{align}
\end{proof}
\end{widetext}

\section{Trivialization of multipartite communication complexity with nontrivial functional boxes}
In Ref. \cite{vanDam} it was shown that the possibility of existence of strong correlations known as PR-boxes \cite{prbox} 
would lead to the trivialization of communication complexity, since by sharing sufficient number of PR-boxes, Alice and Bob would
be able to compute any distributed boolean function with only one bit of communication. 
Later this result was generalized to functional boxes \cite{PRd}, \ie  a generalization of PR-boxes for $d$ outputs, where
the outputs satisfy $a \oplus_d b =f(x,y)$, with $d$ prime and any additively inseparable function $f(x,y)$.
Any functional box which cannot be simulated classically would also lead to a trivialization of 
communication complexity \cite{PRd}. Furthermore, a generalization to binary outcome multi-party
communication complexity scenarios was also considered in Ref. \cite{BP05}. In the multi-party problem, $n$ parties are each 
given an input $x_i$ and their goal is to  compute a function $f(\vec{x})$ of their joint inputs with as little communication as possible. 
In Ref. \cite{BP05}, it was shown that if the parties shared a sufficient number of $n$-party PR boxes, with input-output
relation given by $\bigoplus_{i} a_i = x_1 \cdot \ldots \cdot x_n$, then any $n$-party communication complexity 
problem can be solved with only $n-1$ bits of communication (from $n-1$ parties to the first party who then computes the function), 
thus leading to a trivialization. 

Here we formally derive an analogous result for prime $d$ output multipartite functional boxes, showing that any box which maximally 
saturate an \xor-$d$ game with $n$ players
lead to trivialization of computation of a multi-party function where only $n-1$ dits of communication are necessary.
We start by defining $PR_n$-$d$ boxes, a generalization of $PR$ boxes \cite{prbox} for $n$ parties and prime $d$ outputs.

\begin{definition}
 \begin{align}
  PR_n\text{-}d(\vec{a}|\vec{x})=\begin{cases}
                                          \frac{1}{d^n}, \; \text{if} \; a_1+ \ldots + a_n=x_1 \cdot \ldots \cdot x_n\\
                                          0, \; \text{otherwise}
                                         \end{cases}
 \end{align}
 where $d$ is prime and sum $+$ and multiplication $\cdot$ are operations modulo $d$.
\end{definition}

\begin{theorem}\label{sPRnd}
 If $n$ parties are allowed to share an arbitrary number of $PR_n$-$d$ boxes, any $n$-partite communication
 complexity problem can be solved with only $n-1$ dits of communication. 
\end{theorem}

\begin{proof}
Our proof is a straightforward generalization of Ref. \cite{PRd}.
We prove for the case $n=3$. The proof for general $n$ follows directly.

We start by observing that any function 
$F(\vec{x},\vec{y},\vec{z})$, $F: \mathbb{Z}_d^{m_1} \times  \mathbb{Z}_d^{m_2} \times  \mathbb{Z}_d^{m_3}\rightarrow \mathbb{Z}_d$,
can be written as a multivariate polynomial with degree at most $d-1$ in each $x_i, y_j$ and $z_k$
\begin{align}
F(\vec{x},\vec{y},\vec{z})=\sum_{\vec{\alpha},\vec{\beta},\vec{\gamma}}\mu_{\vec{\alpha},\vec{\beta},\vec{\gamma}}\; {\vec{x}}^{\,\vec{\alpha}} \, \vec{y}^{\,\vec{\beta}} \, \vec{z}^{\,\vec{\gamma}} , 
\end{align}
where ${\vec{x}}^{\,\vec{\alpha}}=\prod_{i=1}^{m_1}x_i^{\alpha_i}$, ${\vec{y}}^{\,\vec{\beta}}=\prod_{j=1}^{m_2}y_j^{\beta_j}$, 
${\vec{z}}^{\,\vec{\gamma}}= \prod_{k=1}^{m_3}z_k^{\gamma_k}$, $\mu_{\vec{\alpha},\vec{\beta},\vec{\gamma}} \in \mathbb{Z}_d$, $\vec{\alpha}=(\alpha_1,\ldots,\alpha_{m_1}) \in \mathbb{Z}_d^{m_1}$
and analogously for $\vec{\beta}$ and $\vec{\gamma}$.
Now if the players have access to $r=d^{m_1   m_2  m_3}$ $PR_3$-$d$ boxes they can execute the following protocol in
order to compute $F(\vec{x},\vec{y},\vec{z})$:
\begin{enumerate}
 \item For each $(\vec{\alpha},\vec{\beta},\vec{\gamma})$ the players picks one $PR_3$-$d$,
 \item Alice inputs ${\vec{x}}^{\,\vec{\alpha}}=\prod_{i=1}^{m_1}x_i^{\alpha_i}$. Bob inputs ${\vec{y}}^{\,\vec{\beta}}=\prod_{j=1}^{m_2}y_j^{\beta_j}$. 
 Charlie inputs
 ${\vec{z}}^{\,\vec{\gamma}}=\prod_{k=1}^{m_3}z_k^{\gamma_k}$. And they get respectively the outputs $a_{\vec{\alpha}}$, $b_{\vec{\beta}}$ and $c_{\vec{\gamma}}$.
 \item Bob sets $b=\sum_{\vec{\alpha}}\sum_{\vec{\beta}}\sum_{\vec{\gamma}}\mu_{\vec{\alpha},\vec{\beta},\vec{\gamma}}b_{\vec{\beta}}$ and send $b$ to Alice.
 Charlie sets $c=\sum_{\vec{\alpha}}\sum_{\vec{\beta}}\sum_{\vec{\gamma}}\mu_{\vec{\alpha},\vec{\beta},\vec{\gamma}}c_{\vec{\gamma}}$ and send $c$ to Alice.
  \item Alice sets $a=\sum_{\vec{\alpha}}\sum_{\vec{\beta}}\sum_{\vec{\gamma}}\mu_{\vec{\alpha},\vec{\beta},\vec{\gamma}}a_{\vec{\alpha}}$ and she
 computes $F(\vec{x},\vec{y},\vec{z})=a+b+c$.
\end{enumerate}
where only 2 dits were communicated in order to compute the function.

The generalization for $n$-party function follows in analogous way, where any function $F:\mathbb{Z}_d^{m_1} \times \ldots \times  \mathbb{Z}_d^{m_n}\rightarrow \mathbb{Z}_d$
will be written as a multivariate polynomial with degree at most $d-1$ in each variable and using an analogous protocol, with $n-1$ parties communicating
only one dit to the first party, the computation of function $F$ will be performed.
\end{proof}

The natural generalization of bipartite functional boxes \cite{PRd} to the multipartite case is the following

\begin{definition}
 For any function $f: \mathbb{Z}_d \times \ldots \times \mathbb{Z}_d \rightarrow \mathbb{Z}_d$, the multipartite functional box corresponding 
 to $f$ is defined as
 \begin{align}
  P_n^f(\vec{a}|\vec{x})=\begin{cases}
                                          \frac{1}{d^n}, \; \text{if} \; a_1+ \ldots + a_n= f(x_1,\ldots,x_n)\\
                                          0, \; \text{otherwise}
                                         \end{cases}
 \end{align}
\end{definition}

Now we argue that all $n$-partite functional box with $f$ non-additively separable ($f$ is additively separable if 
$f(x_1,\ldots,x_n)=f_1(x_1)+f_2(x_2)+\ldots + f_n(x_n)$) would lead to some kind of trivialization of communication complexity.

\begin{theorem} 
 All $P_3^f$ with $f(x,y,z)$ such that there exists a partial derivative of some order equals to $\lambda \cdot x\cdot y\cdot z+g(x)+h(y)+s(z)$ can be 
 used to simulate a $PR_3$-$d$, and then can be used to solve any $3$-partite communication
 complexity problem  with only $2$ dits of communication.
\end{theorem}

\begin{proof}
First let us consider 
\begin{align}\label{sfPR3d}
 f(x,y,z)=\lambda \cdot x\cdot y\cdot z+g(x)+h(y)+s(z).
\end{align}
So by using box $P_3^f$ Alice, Bob and Charlie can input $x,y$ and $z$ respectively
and get outputs $a$, $b$ and $c$. Now following Ref. \cite{PRd} Alice sets $a'=\lambda^{-1}(a-g(x))$, Bob sets $b'=\lambda^{-1}(b-h(y))$ and Charlie sets $c'=\lambda^{-1}(c-s(z))$, 
so that we have
\begin{align}
 a'+b'+c'=x\cdot y \cdot z\,.
\end{align}
In order to randomize the results they can randomly chose $k \in \mathbb{Z}_d$ and output $a_f=a'+k$, $b_f=b'+k$ and $c_f=c'-2k$, so that they perfectly
simulate a $PR_3$-$d$ box.

Now for other functions $f$ we can use the method of Ref. \cite{PRd} of applying partial derivatives to the function. 
The partial derivative of $f$ with respect to $x$ is defined as
\begin{align}
 f_x(x,y,z)\equiv f(x+1,y,z)-f(x,y,z)
\end{align}
and it generates a polynomial with the degree in $x$ reduced by 1, while the degree in $y$ and $z$ remains the same or is smaller. And note that 
with two boxes $P_3^f$ we can simulate the box $ f_x(x,y,z)$.

Then if by partial derivatives we can reduce function $f$ to the form \eqref{sfPR3d} we have a protocol using a finite number of boxes $P_3^f$ to simulate
$PR_3$-$d$. By the result of Theorem \ref{sPRnd}, with an arbitrary finite number of $P_3^f$, we can solve  any $3$-partite communication
 complexity problem with only $2$ dits of communication.
 \end{proof}

If a function $f(x,y,z)$ is not additively separable it will contain at least one term involving product of two variables, for example $x^ry^s$ and this box 
can be reduced, by derivatives, into a box of the form $\lambda \cdot x \cdot y+g(x)+h(y)+s(z)$.
Now using the results for the bipartite case \cite{PRd}, if Charlie always 
inputs $z=1$, with only 2 dits of communication they can compute any function of two variables $f(x,y)$.
 
\section{No quantum realization of functional nonlocal boxes} 

We now show that using our bounds we can exclude the existence of functional boxes that would lead to the trivialization of
communication complexity in a multipartite scenario.

\begin{theorem}\label{stotalfunc}
 For a $n$-player $d$ outcome {\sc{xor}} games $g_n^{\oplus}$ with $m$ questions per player and uniform input distribution
 $p(\vec{x})=1/m^n$, $\omega_q(g_n^{\oplus})=1$ iff $\omega_c(g_n^{\oplus})=1$.
\end{theorem}

\begin{widetext}
\begin{proof} 
We start by proving the result for 3 players. We first chose the partition $S=\DE{A}$ to write the game matrices.

The constraint that the input distribution is $p(x,y,z)=1/m^3$ for all $x,y,z$ implies $\norm{\Phi^A_k}\leq 1/\sqrt{m^3}$ since $\Phi_k^A{\Phi_k^{A}}^{\dagger}$ is an
$m \times m$ matrix with the absolute value of all elements equal to $1/m^4$ (and then $\norm{\Phi^A_k{\Phi_k^{A}}^{\dagger}} \leq 1/m^3$).
Now considering the bound
\begin{align}
  \omega_q(g_n^{\oplus})\leq \frac{1}{d}\De{1+\sqrt{m^3}\sum_{k=1}^{d-1}\norm{\Phi_k^A}}
\end{align}
we see that  $\omega_q(g^{\oplus})=1$ requires $\norm{\Phi_k^A}= 1/\sqrt{m^3}$ for all $k$.

Let $\ket{\lambda^A}=(\lambda^A_1,\ldots,\lambda^A_m)$ be the maximum eigenvector corresponding to eigenvalue $1/{m^3}$ of $\Phi_1^A{\Phi_1^A}^{\dagger}$.
Consider $ \lambda^A_j=|\lambda^A_j|\zeta^{\theta^A_j}$ and assume $|\lambda^A_1| \geq |\lambda^A_2| \geq \ldots \geq |\lambda^A_m|$.

Now we have
\begin{align}
 \Phi^A_1=\sum_{x,y,z}\frac{1}{m^3} \zeta^{f(x,y,z)} \ketbra{x}{yz}
\end{align}
and then
\begin{align}
\Phi^A_1 {\Phi_1^A}^{\dagger} =& \frac{1}{m^6} \sum_{x,y,z} \sum_{x',y',z'}\zeta^{f(x,y,z)-f(x',y',z')}\ket{x}\braket{yz}{y'z'}\bra{x'}\\
 =& \frac{1}{m^6} \sum_{x,x',y,z}\zeta^{f(x,y,z)-f(x',y,z)}\ketbra{x}{x'} \nonumber
\end{align}
and
\begin{align}
\Phi^A_1 {\Phi^A}_1^{\dagger}\ket{\lambda} =&  \frac{1}{m^6} \sum_{x,x',y,z,i}\zeta^{f(x,y,z)-f(x',y,z)+\theta^A_i}|\lambda^A_i|\ket{x}\braket{x'}{i}\\
 =& \frac{1}{m^6} \sum_{x,y,z,i}\zeta^{f(x,y,z)-f(i,y,z)+\theta^A_i}|\lambda^A_i|\ket{x} \nonumber.
\end{align}
Analyzing the first component of the eigenvalue equation ${\Phi^A}_1{\Phi^A}_1^{\dagger}\ket{\lambda^A}=1/m^3 \ket{\lambda^A}$ we have
\begin{align}
 \De{\Phi^A_1{\Phi^A}_1^{\dagger}\ket{\lambda^A}}_1 =& \frac{1}{m^6} \sum_{y,z,i}\zeta^{f(1,y,z)-f(i,y,z)+\theta^A_i}|\lambda^A_i|=\frac{1}{m^3} |\lambda^A_1| \zeta^{\theta^A_1}.
\end{align}
In order to satisfy this equation we need to have 
\begin{subequations}
\begin{align}
 |\lambda^A_i|&=|\lambda^A_1| \; \; \forall \;\;i\\
 &\text{and} \nonumber \\ 
\zeta^{f(1,y,z)-f(i,y,z)+\theta^A_i}&= \zeta^{\theta^A_1}  \; \; \forall \;\;i, y, z.
\end{align}
\end{subequations}
The equations for the other components of the eigenvalue equation imply: 
\begin{subequations}
\begin{align}\label{sthetaA}
 f(x,y,z)-f(x',y,z)=\theta^A_{x}-\theta^A_{x'} \; \forall\; y,z,
\end{align}
where the operations are modulo $d$.

We can do the same argument for the 
other partitions $S=\DE{B}$ and $S=\DE{C}$, and the hypothesis of $\omega_q(g_n^{\oplus})=1$ implies by the 
same arguments above that $\rank (\Phi^S_1)=1$ for all $S$ and then we have the relations:
\begin{align}
f(x,y,z)-f(x,y',z)=\theta^B_{y}-\theta^B_{y'}\; \forall\; x,z\label{sthetaB}\\
f(x,y,z)-f(x,y,z')=\theta^C_{z}-\theta^C_{z'}\; \forall\; x,y\label{sthetaC}.
\end{align}
\end{subequations}

By relations \eqref{sthetaA}, \eqref{sthetaB} and \eqref{sthetaC} we can deduce that
\begin{align}
 f(x,y,z)&=(\theta^A_{x}-\theta^A_{0})+f(0,y,z)\nonumber\\
 &=(\theta^A_{x}-\theta^A_{0})+(\theta^B_{y}-\theta^B_{0})+f(0,0,z)\\
 &=(\theta^A_{x}-\theta^A_{0})+(\theta^B_{y}-\theta^B_{0})+(\theta^C_{z}-\theta^C_{0})+f(0,0,0)\nonumber
\end{align}
and then consider $a_0, b_0, c_0$ such that $a_0\oplus_d b_0 \oplus_d c_0=f(0,0,0)$, the classical strategy
\begin{align}
 a&=a_0+(\theta^A_{x}-\theta^A_{0}) \nonumber\\
  b&=b_0+(\theta^B_{y}-\theta^B_{0})\\
   c&=c_0+(\theta^C_{z}-\theta^C_{0})\nonumber
\end{align}
win the game with probability 1.\vspace{1em}

The proof for a $n$-player game follows in the same way.
Considering the partition $S=\DE{A_1}$.
The constraint of equally distributed inputs implies $\Phi^S_k{\Phi^S}_k^{\dagger}$ is an
$m \times m$ matrix with the absolute value of all elements equal to $1/m^{n+1}$ (and then $\norm{\Phi_k\Phi_k^{\dagger}} \leq 1/m^n$).
Now considering the bound
\begin{align}
  \omega_q \leq \frac{1}{d}\De{1+\sqrt{m^n}\sum_{k=1}^{d-1}\norm{\Phi^S_k}}
\end{align}
we see that  $\omega_q(g^{\oplus})=1$ requires $\norm{\Phi^S_k}= 1/\sqrt{m^n}$ for all $k$.

Analogously for the 3-player game, we conclude that in order to satisfy $\norm{\Phi^S_k}= 1/\sqrt{m^n}$ all the rows of the 
game matrix has to be proportional to each other and then
\begin{align}\label{sthetaA1}
 f(x_1,x_2,\ldots,x_n)-f({x'}_1,x_2,\ldots,x_n)=\theta^{A_1}_{x_1}-\theta^{A_1}_{{x'}_1} \;\; \forall\;\; x_2,\ldots x_n.
\end{align}
Running the analysis over the partitions $S={A_i}$ for $i=\DE{2,\ldots,n}$ we can specify a classical strategy that wins the 
game with probability 1. 
\end{proof}
\end{widetext}

\section{Device independent witnesses of genuine tripartite entanglement: Explicit example and numerical results}

Based on the Mermin inequality \cite{Mermin} for the GHZ-paradox, we have considered the game:
\begin{align}\label{sgameghz}
a \oplus_3 b \oplus_3 c=x\cdot y \cdot z \;\;\; \text{s.t.}\;\;\; x\oplus_3 y \oplus_3 z=0
\end{align}
where $x, y, z, a,b,c \in \DE{0,1,2}$, and  $\oplus_3, \cdot$ are sum and multiplication modulo 3.

Bellow we present the projective measurements  $\{M_x^a\}=\DE{ \ketbra{A_x^a}{A_x^a}}$, $\{M_y^b\}=\DE{\ketbra{B_y^b}{B_y^b}}$, 
$\{M_z^c\}=\DE{ \ketbra{C_z^c}{C_z^c}}$ 
that allow the players to win game \eqref{sgameghz} with probability 1 when they share the $GHZ_3$ state:

 \begin{widetext}
 \begin{subequations}
 \label{sproj}
\begin{align}
 \ket{A_0^0}=\ket{B_0^0}=\frac{1}{\sqrt{3}}\de{\,1\,,\zeta^{4/3},\,1\,}\;,\; \ket{A_0^1}=\ket{B_0^1}=\frac{1}{\sqrt{3}}\de{\zeta^{2},\zeta^{7/3},\,1\,} \;,\; \ket{A_0^2}=\ket{B_0^2}=\frac{1}{\sqrt{3}}\de{\zeta,\zeta^{1/3},\,1\,}\label{A0}\\
\ket{A_1^0}=\ket{B_1^0}=\frac{1}{\sqrt{3}}\de{\zeta^{1/3},\,1\,,\,1\,}\;,\; \ket{A_1^1}=\ket{B_1^1}=\frac{1}{\sqrt{3}}\de{\zeta^{7/3},\zeta,\,1\,} \;,\; \ket{A_1^2}=\ket{B_1^2}=\frac{1}{\sqrt{3}}\de{\zeta^{4/3},\zeta^{2},\,1\,}\label{A1}\\
\ket{A_2^0}=\ket{B_2^0}=\frac{1}{\sqrt{3}}\de{\zeta^{8/3},\zeta^{8/3},1}\;,\; \ket{A_2^1}=\ket{B_2^1}=\frac{1}{\sqrt{3}}\de{\zeta^{5/3},\zeta^{2/3},1} \;,\; \ket{A_2^2}=\ket{B_2^2}=\frac{1}{\sqrt{3}}\de{\zeta^{2/3},\zeta^{5/3},1}\label{A2}
 \end{align}

 \begin{align}
  \ket{C_0^0}=\frac{1}{\sqrt{3}}\de{\,1\,,\zeta^{1/3},\,1\,}\;,\; \ket{C_0^1}=\frac{1}{\sqrt{3}}\de{\zeta^{2},\zeta^{4/3},\,1\,} \;,\; \ket{C_0^2}=\frac{1}{\sqrt{3}}\de{\zeta,\zeta^{7/3},\,1\,}\label{C0}\\
\ket{C_1^0}=\frac{1}{\sqrt{3}}\de{\zeta^{1/3},\zeta^2,,\,1\,}\;,\; \ket{C_1^1}=\frac{1}{\sqrt{3}}\de{\zeta^{7/3},\,1\,,\,1\,} \;,\; \ket{C_1^2}=\frac{1}{\sqrt{3}}\de{\zeta^{4/3},\zeta,\,1\,}\label{C1}\\
\ket{C_2^0}=\frac{1}{\sqrt{3}}\de{\zeta^{8/3},\zeta^{5/3},1}\;,\; \ket{C_2^1}=\frac{1}{\sqrt{3}}\de{\zeta^{5/3},\zeta^{8/3},1} \;,\; \ket{C_2^2}=\frac{1}{\sqrt{3}}\de{\zeta^{2/3},\zeta^{2/3},1}\label{C2}
 \end{align}
where $\zeta=e^{2\pi i/3}$.
\end{subequations}
 \end{widetext}

 Note that the ``observable'' defined as $ A_x^i=\sum_a\bar{\chi}_i(a)M_x^a,$ has trace zero for the projective measurements specified above, and 
 analogously for $B_y^j$ and $C_z^k$, then we have that 
 \begin{align}
  \Tr\de{I(A_x^i\otimes B_y^j \otimes C_z^k)}=0.
 \end{align}

 Hence we can easily calculate the success probability for the noisy $GHZ_3$ state when the players perform the local measurements given in \eqref{sproj}:
 
 \begin{widetext}
\begin{align}
  \omega(\rho(V))& = \frac{1}{3}\de{1+\sum_{x,y,z}\sum_{k=1}^2 {p(x,y,z)} \zeta^{k\cdot f(x,y,z)}\Tr \de{\rho(V)(A_x^k\otimes B_y^k\otimes C_z^k)}}\nonumber\\
		&= \frac{1}{3}\de{1+\sum_{x,y,z}\sum_{k=1}^2 {p(x,y,z)} \zeta^{k\cdot f(x,y,z)}V \Tr \de{\ketbra{GHZ_3}{GHZ_3}(A_x^k\otimes B_y^k\otimes C_z^k)}}\\
		&=\frac{1-V}{3}+V\omega(GHZ_3)\nonumber\\
		&=\frac{1+2V}{3} \nonumber.
 \end{align}

\end{widetext}

\end{document}